\documentclass[final,a4paper]{siamltex}
\usepackage{amsmath,amssymb,amsfonts}
\usepackage{verbatim}
\usepackage{changepage}
\usepackage{bm,color,url}
\usepackage{algorithm,algorithmic}
\usepackage{geometry}
\usepackage{setspace}
\usepackage[notref,notcite]{showkeys}
\usepackage{hyperref}

\newtheorem{Lemma*}{theorem}
\newtheorem{remark}{Remark}[section]
\newtheorem{remark*}{Remark}

\newcommand{\bZ}{{\bf Z}}
\newcommand{\bI}{\mbox{\bf I}}

\def\Ebb{\mathbb{E}}

\title{Convergence Analysis of Schr{\"o}dinger-F{\"o}llmer Sampler without Convexity}
\author{Yuling Jiao
\thanks{School of Mathematics and Statistics, Wuhan University, Wuhan, China.
Email: yulingjiaomath@whu.edu.cn} \quad \and Lican Kang
\thanks{School of Mathematics and Statistics, Wuhan University, Wuhan, China.
Email: kanglican@whu.edu.cn}
\quad\and
Yanyan Liu
\thanks{School of Mathematics and Statistics, Wuhan University, Wuhan, China.
Email: liuyy@whu.edu.cn}
\quad\and Youzhou Zhou \thanks{ Department of Pure mathematics, Xi'an Jiaotong-Liverpool University, Suzhou, China.
Email: youzhou.zhou@xjtlu.edu.cn}
}

\begin{document}

\maketitle

\begin{abstract}
Schr\"{o}dinger-F\"{o}llmer sampler (SFS) \cite{sfs21},  is a novel  and efficient approach for sampling from possibly unnormalized distributions without ergodicity.  SFS is based on the Euler-Maruyama discretization of Schr\"{o}dinger-F\"{o}llmer diffusion process
$$\mathrm{d} X_{t}=-\nabla U\left(X_t, t\right) \mathrm{d} t+\mathrm{d} B_{t}, \quad t \in[0,1],\quad X_0=0$$
 on the unit interval, which transports the degenerate distribution at time zero to the target distribution at time one.
In \cite{sfs21}, the consistency of SFS is established under a restricted assumption that  
{the potential $U(x,t)$
is uniformly  (on $t$) strongly 
convex
(on $x$).
} In this paper we provide   a nonasymptotic error bound of SFS in  Wasserstein distance under  some smooth  and bounded conditions on the
 density ratio of the target distribution over the standard normal distribution,  but   without requiring   the strongly
{convexity of the
 potential.
}
\end{abstract}

\section{Introduction}\label{introduction}
Sampling from  possibly unnormalized distributions  is an  import task in Bayesian statistics and machine learning.
 Ever since the Metropolis-Hastings (MH) algorithm \cite{metropolis1953equation,hastings1970monte} was introduced, various random sampling methods were proposed, including
Gibbs sampler,  random walk sampler, independent sampler, Lagevin sampler,   bouncy particle sampler, zig-zag sampler  \cite{geman1984stochastic,gelfand1990sampling,Tierney1994Markov,liu2008monte,
robert2010introducing,bouchard2018bouncy,bierkens2019zig}, among others, see \cite{brooks2011handbook,dunson2020hastings}  and the references therein. The above mentioned sampling algorithms generate random samples   by running an ergodic Markov chain whose stationary distribution is the target distribution.

In \cite{sfs21}, the  Schr\"{o}dinger-F\"{o}llmer sampler (SFS),  a novel sampling approach  without requiring the property of ergodicity is proposed. SFS is  based on  the  Schr\"{o}dinger-F\"{o}llmer diffusion process, defined as
\begin{align}\label{sde}
\mathrm{d} X_{t}=b\left(X_{t}, t\right) \mathrm{d} t+\mathrm{d} B_{t}, \quad t \in[0,1],\quad X_0=0,
\end{align}
where the drift function $$b(x,t) = -\nabla U(x,t)=\frac{\Ebb_{Z\sim N\left(0,\bI_{p}\right)}[\nabla f(x+\sqrt{1-t}Z)]}{\Ebb_{Z\sim N\left(0,\bI_{p}\right)}[f(x+\sqrt{1-t} Z)]}:\mathbb{R}^p \times [0,1]\rightarrow \mathbb{R}^1$$
 with $f(\cdot)=\frac{d \mu }{d N\left(0,\bI_{p}\right)}(\cdot)$.
According to \cite{leonard2014survey} and \cite{eldan2020}, the process  $\{X_t\}_{t\in [0,1]}$  in (\ref{sde}) was first formulated by F\"{o}llmer \cite{follmer1985, follmer1986, follmer1988} when studying the Schr\"{o}dinger bridge problem \cite{schrodinger1932theorie}.
The main feature of the above  Schr\"{o}dinger-F\"{o}llmer  process is that  it interpolates $\delta_{0}$ and the target $\mu$ in time $[0,1]$, i.e.,  $X_1 \sim \mu$, see Proposition \ref{SBP}.  SFS  samples from $\mu$ via the  following Euler-Maruyama discretization of (\ref{sde}),
$$
Y_{t_{k+1}}=Y_{t_k}+ sb\left(Y_{t_k}, t_k\right)+\sqrt{s}\epsilon_{k+1},
~
Y_{t_0} = 0,~ k=0,1,\ldots, K-1,
$$
where $s = 1/K$ is the  step size, $t_k = s k$,  and $\{\epsilon_{k}\}_{k=1}^{K}$ are independent and identically distributed from $N(0,\bI_{p})$.
 If the expectations in the   drift term $b(x,t)$   do not have  analytical forms,  one can use   Monte Carlo  method   to evaluate
  $b\left(Y_{t_k}, t_k\right)$ approximately, i.e., one can sample from $\mu$ according
 $$
\widetilde{Y}_{t_{k+1}}=\widetilde{Y}_{t_k}+ s\tilde{b}_m\left(\widetilde{Y}_{t_k}, t_k\right)+\sqrt{s}\epsilon_{k+1},
~
\widetilde{Y}_{t_0} = 0, ~k=0,1,\ldots, K-1,$$
  where $\tilde{b}_m(\widetilde{Y}_{t_{k}},t_{k})=
\frac{\frac{1}{m}\sum_{j=1}^m[\nabla f(\widetilde{Y}_{t_{k}}+\sqrt{1-t_{k}}Z_j)]}{\frac{1}{m}\sum_{j=1}^{m} [f(\widetilde{Y}_{t_{k}}+\sqrt{1-t_{k}} Z_j)]}$ with $Z_1,...Z_m$ i.i.d $N(0,\bI_{p})$. The numerical simulations in   \cite{sfs21} demonstrate  that SFS outperforms the exiting samplers based on
egordicity.

In Section 4.2 of \cite{sfs21}, they prove that
$$W_2(\mbox{Law}(\widetilde{Y}_{t_K}),\mu)\rightarrow 0, \ \ \mathrm{as} \ \  s \rightarrow 0, m \rightarrow \infty$$ under a restricted assumption that
{the potential $U(x,t)$ is uniformly strongly convex, i.e.,
\begin{equation}\label{cond4}
U(x,t)-U(y,t)-\nabla U(y,t)^{\top}(x-y)\geq (M/2)\left\|x-y\right\|^2_2, \forall x,y\in \mathbb{R}^p,  \forall t\in [0,1],
\end{equation}
where $M$ is one finite and positive constant.
}
In this paper we provide a new analysis of the above SFS iteration. We establish a nonasymptotic error bound on
$W_2(\mbox{Law}(\widetilde{Y}_{t_K}),\mu)$ under   the condition that  $f$ and $\nabla f$ are Lipschitz continuous and $f$ has positive  lower bound,     but without using  the uniformly strongly convexity requirement (\ref{cond4}).

The rest of this paper is organized as follows. In Section \ref{method},  we  recall the SFS  method.
In Section \ref{Theorey}, we present our  theoretical analysis.
We conclude in Section \ref{conlusion}. Proofs for all the theorems are provided in Appendix \ref{append}.

\section{Schr{\"o}dinger-F{\"o}llmer sampler}\label{method}
In this section we  recall the Schr{\"o}dinger-F{\"o}llmer  sampler briefly. More   backgrounds  on the  Schr{\"o}dinger-F{\"o}llmer diffusion process please see \cite{dai1991stochastic,leonard2014survey,chen2021stochastic,sfs21}.

Let $\mu \in \mathcal{P}\left(\mathbb{R}^{p}\right)$ be the target distribution and absolutely continuous with respect to the $p$-dimensional standard normal measure $G=N(0,\bI_{p})$. Let $$f(x)=\frac{d\mu}{dG}(x).$$
We assume that
\begin{itemize}
\item[$(\textbf{A1})$] $f,\nabla f$ are Lipschitz continuous with constant $\gamma$,
\item[$(\textbf{A2})$] There exists $\xi>0$ such that $f\geq \xi$.
\end{itemize}
Define the heat semigroup $Q_t, t\in [0,1]$ as
$$
Q_{t} f(x)=\Ebb_{Z \sim G}[f(x+\sqrt{t} Z)].
$$

\begin{proposition}\label{SBP}
Define a drift function
$$
b(x, t)=\nabla \log Q_{1-t} f(x).
$$
If $f$ satisfies assumptions $(\textbf{A1})$ and $(\textbf{A2})$, then the Schr{\"o}dinger-F\"{o}llmer diffusion
\begin{align}\label{sch-equation}
\mathrm{d} X_{t}=b\left(X_{t}, t\right) \mathrm{d} t+\mathrm{d} B_{t}, \quad t \in[0,1], \quad X_{0}=0,
\end{align}
has a unique strong solution and $X_{1} \sim \mu$.
\end{proposition}
\begin{remark}
\begin{itemize}
{\item[(i)]
From the definition of $b(x,t)$ in Proposition \ref{SBP}, it follows that
$U(x,t)=-\log Q_{1-t}f(x)$.
}
\item[(ii)]
If the target distribution is $\mu(dx)=\exp (-V(x)) d x / C$ with the normalized constant $C$,
then $f(x)=\frac{\left(\sqrt{2\pi}\right)^p}{C}\exp (-V(x)+\frac{\|x\|_2^2}{2})$.
Once $V(x)$ is twice differentiable and
$$
\lim_{R\to\infty}\sup_{\|x\|_2\geq R}\exp \left(-V(x)+\|x\|_2^{2} / 2\right)\|x-\nabla V\|_2<\infty,
$$
$$
\lim_{R\to\infty}\sup_{\|x\|_2\geq R}\exp \left(-V(x)+\|x\|_2^{2} / 2\right)\|Hess(V)\|_2<\infty,
$$
then both $f$ and $\nabla f$ are Lipschitz continuous, i.e., (\textbf{A1}) holds.
(\textbf{A2}) is equivalent to the growth condition on the potential that  $V(x) \leq \frac{\|x\|^2}{2} - \log \xi + \mathrm{contant}$.
\item[(iii)] Under (\textbf{A1}) and (\textbf{A2}), some calculation shows that  $$\|\nabla f\|_2\leq \gamma, \|Hess(f)\|_2\leq\gamma,$$ and
$$
\sup_{x\in\mathbb{R}^p,t\in[0,1]}\|\nabla Q_{1-t}f(x)\|_2\leq \gamma,
\sup_{x\in\mathbb{R}^p,t\in[0,1]}\|Hess(Q_{1-t}f(x))\|_2\leq \gamma,
$$ and
$$
b(x,t)=\frac{\nabla Q_{1-t}f(x)}{Q_{1-t}f(x)}, ~\nabla b(x,t)=\frac{Hess(Q_{1-t}f)(x)}{Q_{1-t}f(x)}-b(x,t)b(x,t)^{\top}.
$$
We conclude that
$$
\sup_{x\in\mathbb{R}^p,t\in[0,1]}\|b(x,t)\|_2\leq \frac{\gamma}{\xi},
\sup_{x\in\mathbb{R}^p,t\in[0,1]}\|\nabla b(x,t)\|_2\leq \frac{\gamma}{\xi}+\frac{\gamma^2}{\xi^2}.
$$
\end{itemize}
\end{remark}
Proposition   \ref{SBP} shows  that the Schr{\"o}dinger-F\"{o}llmer diffusion will transport $\delta_0$ to the target $\mu $ on the unite time interval. 
Since  drift term $b(x,t)$ is scale-invariant with respect to $f$ in the sense that
$b(x, t)=\nabla \log Q_{1-t} Cf(x), \forall C>0$. Therefore, the Schr{\"o}dinger-F\"{o}llmer diffusion can be used for sampling from $\mu(dx)=\exp (-V(x)) d x/C$,  where  the normalizing constant of  $C$ may not to be known.
To this end, we use the  Euler-Maruyama method to discretize  the  Schr{\"o}dinger-F\"{o}llmer diffusion (\ref{sch-equation}).
Let $$t_{k}=k\cdot s, \ \ ~k=0,1,  \ldots, K, \ \ ~\mbox{with}~ \ \ s=1 / K, \ \ Y_{t_{0}}=0,$$
the Euler-Maruyama scheme reads
\begin{align}\label{emd}
Y_{t_{k+1}}=Y_{t_k}+s b\left(Y_{t_{k}}, t_{k}\right)+\sqrt{s}\epsilon_{k+1},
~
k=0,1,\ldots, K-1,
\end{align}
where   $\{\epsilon_{k}\}_{k=1}^{K}$ are i.i.d. $N(0,\bI_{p})$ and
\begin{align}\label{driftb}
b(Y_{t_{k}},t_{k})=\frac{\Ebb_{Z}[\nabla f(Y_{t_{k}}+\sqrt{1-t_{k}}Z)]}{\Ebb_{Z}[f(Y_{t_{k}}+\sqrt{1-t_{k}} Z)]}=\frac{\Ebb_{Z}[Z f(Y_{t_{k}}+\sqrt{1-t_{k}}Z)]}{\Ebb_{Z}[f(Y_{t_{k}}+\sqrt{1-t_{k}} Z)]\sqrt{1-t_k}}.
\end{align}
From the definition of $b(Y_{t_{k}},t_{k})$  in (\ref{driftb})  we may not get its explicit expression.
Here, we can get one estimator $\tilde{b}_{m}$ of $b$ by replacing $\Ebb_{Z}$ in $b$ with $m$-sample mean, i.e.,
\begin{align}\label{drifte1}
\tilde{b}_m(Y_{t_{k}},t_{k})=
\frac{\frac{1}{m}\sum_{j=1}^m[\nabla f(Y_{t_{k}}+\sqrt{1-t_{k}}Z_j)]}{\frac{1}{m}\sum_{j=1}^{m} [f(Y_{t_{k}}+\sqrt{1-t_{k}} Z_j)]}, \ k=0, \ldots, K-1,
\end{align}
or
\begin{align}\label{drifte2}
\tilde{b}_m(Y_{t_{k}},t_{k})=
\frac{\frac{1}{m}\sum_{j=1}^m[Z_j f(Y_{t_{k}}+\sqrt{1-t_{k}}Z_j)]}{\frac{1}{m}\sum_{j=1}^{m} [f(Y_{t_{k}}+\sqrt{1-t_{k}} Z_j)]\cdot\sqrt{1-t_k}}, \ k=0, \ldots, K-1,
\end{align}
where  $Z_1, \ldots, Z_{m}$ are i.i.d. $N(0, \bI_p)$.
The detailed description of SFS is summarized in following Algorithm \ref{alg:1} below, which is  Algorithm 2 in \cite{sfs21}.
\begin{algorithm}[H]
	\caption{SFS for $\mu = \exp (-V(x))/C$  with  Monte Carlo estimation of the drift term}
    \label{alg:1}
	\begin{algorithmic}[1]
\STATE Input:  $m$, $K$.  Initialize $s=1/K$, $\widetilde{Y}_{t_0}=0$.
\FOR{$k= 0,1,\ldots, K-1$ }
\STATE Sample $\epsilon_{k+1}\sim N(0,\bI_{p})$.
\STATE Sample $Z_i, i,\ldots,m$, from $ N(0,\bI_{p})$.
\STATE Compute $\tilde{b}_{m}$ according to (\ref{drifte1}) or (\ref{drifte2}),
\STATE $
\widetilde{Y}_{t_{k+1}}=\widetilde{Y}_{t_{k}}+s\tilde{b}_{m}\left(\widetilde{Y}_{t_{k}}, t_{k}\right)+\sqrt{s} \epsilon_{k+1}$.
\ENDFOR
\STATE Output:  $\{\widetilde{Y}_{t_k}\}_{k=1}^{K}$.
\end{algorithmic}
\end{algorithm}

In Section 4.2 of \cite{sfs21}, they proved that
$$W_2(\mbox{Law}(\widetilde{Y}_{t_K}),\mu)\rightarrow 0, \ \ \mathrm{as} \ \  s \rightarrow 0, m \rightarrow \infty$$ under a restricted assumption that
the potential is uniformly strongly convex, i.e, $U(x,t)$ satisfies
(\ref{cond4}).
However, (\ref{cond4}) is not easy to verify.
In the next section, we establish a nonasymptotic bound on the  Wasserstein distance between  the  law of  $\widetilde{Y}_{t_K}$  generated by SFS (Algorithm \ref{alg:1}) and the target $\mu$ under smooth and bounded  conditions (\textbf{A1}) and (\textbf{A2}) but
 without using
{the strongly uniform convexity assumption (\ref{cond4}) on  $U(x,t)$.
}
\section{Nonasymptotic Bound  without convexity}\label{Theorey}
Under conditions (\textbf{A1}) and (\textbf{A2}), one can easily deduce the growth condition and Lipschitz/H{\"o}lder continuity of the drift term $b(x,t)$ \cite{sfs21}, i.e.,
\begin{align}\label{cond1}
\|b(x,t)\|_2^2\leq C_0(1+\|x\|_2^2), \tag{C1}
\end{align}
and
\begin{align}\label{cond2}
\|b(x,t)-b(y,t)\|_2\leq C_1 \|x-y\|_2 \tag{C2},
\end{align}
and
\begin{align}\label{cond3}
\|b(x,t)-b(y,s)\|_2\leq  C_1 \left(\|x-y\|_2+|t-s|^{\frac{1}{2}}\right), \tag{C3}
\end{align}
where $C_0$ and $C_1$ are two finite and positive constants.
\begin{remark}\label{R3}
(\ref{cond1}) and (\ref{cond2}) are the essentially sufficient
conditions such that  the
Schr{\"o}dinger-F\"{o}llmer SDE (\ref{sch-equation})
admits the unique strong solution.
 (\ref{cond3}) has  been introduced in  Theorem 4.1 of \cite{tzen2019theoretical}, and it is also similar to the condition H2 of \cite{chau2019stochastic} and Assumption 3.2 of \cite{barkhagen2018stochastic}.
 Obviously,  (\ref{cond3}) implies  (\ref{cond2}),  and (\ref{cond1}) holds if the drift term $b(x,t)$ is bounded over $\mathbb{R}^p\times[0,1]$. 
\end{remark}

  Let $\mathcal{D}(\nu_1, \nu_2)$ be the collection of coupling probability measures on $\left(\mathbb{R}^{2p},\mathcal{B}(\mathbb{R}^{2p})\right)$ such that its respective marginal distributions are $\nu_1$ and $\nu_2$.
The Wasserstein of order $d \geq 1$ with which we measure the discrepancy between $\mbox{Law}(\widetilde{Y}_{t_K})$  and $\mu$  is defined as
\begin{align*}
W_{d}(\nu_1, \nu_2)=\inf _{\nu \in \mathcal{D}(\nu_1, \nu_2)}\left(\int_{\mathbb{R}^{p}} \int_{\mathbb{R}^{p}}\left\|\theta_1-\theta_2\right\|_2^{d} \mathrm{d} \nu\left(\theta_1,\theta_2\right)\right)^{1/d}.
\end{align*}
\begin{theorem}\label{th1}
Assume (\textbf{A1}) and  (\textbf{A2}) hold, then
\begin{align*}
W_2(\mbox{Law}(\widetilde{Y}_{t_K}),\mu)\leq \mathcal{O}(\sqrt{ps})+\mathcal{O}\left(\sqrt{\frac{p}{\log(m)}}\right),
\end{align*}
where $s=1/K$ is the step size.
\end{theorem}

\begin{remark}
This theorem provides some guidance on the selection of $s$ and $m$. To ensure convergence of the distribution of $\widetilde{Y}_{t_K}$, we should set the step size $s = o(1/p)$ and
$m=\exp(p/o(1))$.
In high-dimensional models with a large $p$, we need to generate a large number of random vectors from $N(0, \bI_p)$ to obtain an accurate estimate of the drift term $b$.
If  we assume that $f$ is bounded above 
we can improve the nonasymptotic error bound, in which
$\mathcal{O}\left(\sqrt{p/\log(m)}\right)$
can be improved to be
$\mathcal{O}\left(\sqrt{p/m}\right)$.
\end{remark}
\begin{theorem}\label{th2}
Assume that, in addition to the conditions of Theorem \ref{th1}, $f$ has  a finite upper bound, then
\begin{align*}
W_2(\mbox{Law}(\widetilde{Y}_{t_K}),\mu)\leq \mathcal{O}(\sqrt{ps})+\mathcal{O}\left(\sqrt{\frac{p}{m}}\right),
\end{align*}
where $s=1/K$ is the step size.
\end{theorem}

\begin{remark}
With the boundedness condition on $f$, to ensure convergence of the sampling distribution, we can set the step size $s = o(1/p)$ and
$m=p/o(1)$.
Note that the sample size requirement for approximating the drift term is significantly less
stringent than that in Theorem \ref{th1}.
\end{remark}
\begin{remark}
Langevin sampling method has been studied under  the (strongly) convex potential assumption
\cite{durmus2016high-dimensional,durmus2016sampling,durmus2017nonasymptotic,dalalyan2017further,
dalalyan2017theoretical,cheng2018convergence,dalalyan2019user-friendly};
the dissipativity condition for the drift term \cite{raginsky2017non,mou2019improved,zhang2019nonasymptotic};
  the local convexity condition for the potential function outside a ball \cite{durmus2017nonasymptotic,cheng2018sharp,ma2019sampling,bou2020coupling}.
However, these conditions may not hold for models with multiple modes, for example,   Gaussian mixtures, where their potentials are not convex and the log Sobolev inequality may  not be satisfied.
Moreover, the constant in the log Sobolev inequality depends on the dimensionality exponentially \cite{wang2009log,hale2010asymptotic,menz2014poincare,raginsky2017non}, implying that the  Langevin samplers suffers from the curse of dimensionality.
 SFS does not require the underlying Markov process to be ergodic, therefore, our results in Theorem \ref{th1} and \ref{th2}
 are established under  the smooth and bounded  assumptions (\textbf{A1}) and (\textbf{A2}) on $f$
  but do not  need the above mentioned  conditions used in the analysis of Langevin samplers.
\end{remark}

In Theorem \ref{th1} and Theorem \ref{th2}, we use  (\textbf{A2}), i.e,  $f$ has positive lower bound,
however, (\textbf{A2}) may  not hold if the target distribution admits compact support. To circumvent this difficulty,  we consider the regularized probability  measure $$\mu_{\varepsilon}=(1-\varepsilon)\mu+\varepsilon G, \ \ \varepsilon\in(0,1).$$ The corresponding density ratio is
$$f_{\varepsilon} = \frac{d \mu_{\varepsilon}}{d G}=(1-\varepsilon)f+\varepsilon.$$ Obviously,
$f_{\varepsilon}$ satisfies (\textbf{A1}) and (\textbf{A2}) if   $f$ and $\nabla f$ are Lipschitz continuous.   Since $\mu_{\varepsilon}$ can approximate to $\mu$ well if we set $\varepsilon$ small enough, then
we consider sampling from $\mu_{\varepsilon}$ by running  SFS (Algorithm \ref{alg:1}) with $f$ being replaced by  $f_{\varepsilon}$.
We use $\widetilde{Y}_{t_K}^{\varepsilon}$ to denote  the last iteration of SFS.
\begin{theorem}\label{th3}
Assume (\textbf{A1}) holds and set $\varepsilon=(\log(m))^{-1/5}$, then
\begin{align*}
W_2(\mbox{Law}(\widetilde{Y}_{t_K}^{\varepsilon}),\mu)\leq \mathcal{O}(\sqrt{ps})+\widetilde{C}_p\cdot\mathcal{O}\left(\frac{1}{(\log(m))^{1/10}}\right),
\end{align*}
where  $s=1/K$ is the step size, $\widetilde{C}_p$ is a constant depending on $p$.
Moreover, if $f$ has the finite upper bound and set $\varepsilon=m^{-1/5}$, then
\begin{align*}
W_2(\mbox{Law}(\widetilde{Y}_{t_K}^{\varepsilon}),\mu)\leq \mathcal{O}(\sqrt{ps})+\widetilde{C}_p\cdot\mathcal{O}\left(\frac{1}{m^{1/10}}\right).
\end{align*}
\end{theorem}
\section{Conclusion}\label{conlusion}

In \cite{sfs21}, Schr\"{o}dinger-F\"{o}llmer sampler (SFS) was proposed  for sampling from possibly unnormalized distributions. The key feature of SFS is that it does not need ergodicity as its theoretical basis.
The consistency of SFS proved in \cite{sfs21} relies a restricted assumption that  
the potential function is uniformly  strongly convex.
In this paper we provide a new convergence  analysis of the SFS  without the  strongly  convexity condition on the potential.
 We establish a nonasymptotic error bound on Wasserstein distance between the law of the output of SFS and the target distribution
under   smooth and bounded assumptions on the   density ratio of the target distribution  over the standard normal distribution.

\section{Acknowledgment}
The authors would like to thank Professor Liming Wu at Universit\'{e} Clermont-Auvergne
for helpful discussions on this topic.

\begin{appendix}
\section{Appendix} \label{append}
\setcounter{equation}{0}
\def\theequation{A.\arabic{equation}}
In this appendix, we prove
Proposition  \ref{SBP} and Theorems \ref{th1}-\ref{th3}.

\subsection{Proof of Proposition  \ref{SBP}}
\begin{proof}
This is a known result, see for example \cite{dai1991stochastic,lehec2013representation}.
\end{proof}
\subsection{Preliminary lemmas for Theorems \ref{th1}-\ref{th2}}
First, we introduce  Lemmas \ref{lemma2}-\ref{lemma6} in preparing for the
proofs of Theorems \ref{th1}-\ref{th2}.
\begin{lemma}\label{lemma2}
Assume (\textbf{A1}) and (\textbf{A2}) hold,
then
\begin{align*}
\Ebb[\|X_t\|_2^2]\leq 2(C_0+p)\exp(2C_0t).
\end{align*}
\end{lemma}
\begin{proof}
From the definition of $X_t$ in (\ref{sch-equation}), we have
$
\|X_t\|_2\leq \int_{0}^t\|b(X_u,u)\|_2\mathrm{d}u+\|B_t\|_2.
$
Then, we can get
\begin{align*}
\|X_t\|_2^2&\leq
2\left(\int_{0}^t\|b(X_u,u)\|_2\mathrm{d}u\right)^2+2\|B_t\|_2^2\\
&\leq
2t\int_{0}^t\|b(X_u,u)\|_2^2\mathrm{d}u+2\|B_t\|_2^2\\
&\leq
2t\int_{0}^tC_0[\|X_u\|_2^2+1]\mathrm{d}u+2\|B_t\|_2^2,
\end{align*}
where the first inequality holds by the inequality $(a+b)^2\leq 2a^2+2b^2$, the last inequality holds by (\ref{cond1}).
Thus,
\begin{align*}
\Ebb\|X_t\|_2^2&\leq
2t\int_{0}^tC_0(\Ebb\|X_u\|_2^2+1)\mathrm{d}u+2\Ebb\|B_t\|_2^2\\
&\leq
2C_0\int_{0}^t \Ebb\|X_u\|_2^2\mathrm{d}u+2(C_0+p).
\end{align*}
By Bellman-Gronwall inequality, we have
\begin{align*}
\Ebb\|X_t\|_2^2\leq 2(C_0+p)\exp(2C_0t).
\end{align*}
\end{proof}
\begin{lemma}\label{lemma3}
Assume (\textbf{A1}) and (\textbf{A2}) hold,
then for any $0\leq t_1\leq t_2\leq 1$,
\begin{align*}
\Ebb[\|X_{t_2}-X_{t_1}\|_2^2]\leq 4C_0\exp(2C_0)(C_0+p)(t_2-t_1)^2+2C_0(t_2-t_1)^2+2p(t_2-t_1).
\end{align*}
\end{lemma}
\begin{proof}
From the definition of $X_t$ in (\ref{sch-equation}), we have
\begin{align*}
\|X_{t_2}-X_{t_1}\|_2\leq \int_{t_1}^{t_2}\|b(X_u,u)\|_2\mathrm{d}u+\|B_{t_2}-B_{t_1}\|_2.
\end{align*}
Then, we can get
\begin{align*}
\|X_{t_2}-X_{t_1}\|_2^2&\leq
2\left(\int_{t_1}^{t_2}\|b(X_u,u)\|_2\mathrm{d}u\right)^2+2\|B_{t_2}-B_{t_1}\|_2^2\\
&\leq
2(t_2-t_1)\int_{t_1}^{t_2}\|b(X_u,u)\|_2^2\mathrm{d}u+2\|B_{t_2}-B_{t_1}\|_2^2\\
&\leq
2(t_2-t_1)\int_{t_1}^{t_2}C_0[\|X_u\|_2^2+1]\mathrm{d}u+2\|B_{t_2}-B_{t_1}\|_2^2,
\end{align*}
where the last inequality holds by (\ref{cond1}).
Hence,
\begin{align*}
\Ebb\|X_{t_2}-X_{t_1}\|_2^2&\leq
2(t_2-t_1)\int_{t_1}^{t_2}C_0(\Ebb\|X_u\|_2^2+1)\mathrm{d}u+2\Ebb\|B_{t_2}-B_{t_1}\|_2^2\\
&\leq
4C_0\exp(2C_0)(C_0+p)(t_2-t_1)^2+2C_0(t_2-t_1)^2+2p(t_2-t_1),
\end{align*}
where the last inequality holds by Lemma \ref{lemma2}.
\end{proof}
\begin{lemma}\label{lemma4}
Assume (\textbf{A1}) and (\textbf{A2}) hold,
then for any $R>0$,
\begin{align*}
\underset{\|x\|_2
\leq R,t\in [0,1]}{\sup} \Ebb\left[\|b(x,t)-\tilde{b}_m(x,t)\|_2^2\right]
\leq \mathcal{O}\left(\frac{p\exp(R^2)}{m}\right).
\end{align*}
Moreover, if $f$ has the finite upper bound, then
\begin{align*}
\underset{x\in\mathbb{R}^p,t\in [0,1]}{\sup} \mathbb{E}\left[\|b(x,t)-\tilde{b}_m(t,x)\|_2^2\right]
\leq \mathcal{O}\left(\frac{p}{m}\right).
\end{align*}
\end{lemma}
\begin{proof}
Denote two independent sets of independent copies of $Z\sim N(0, \bI_p)$,
that is, $\bZ=\{Z_1,\ldots,Z_m\}$ and $\bZ^{\prime}=\{Z_1^{\prime},\ldots,Z_m^{\prime}\}$.
For notation convenience, we  denote
\begin{align*}
&d=\Ebb_{Z}\nabla f(x+\sqrt{1-t}Z),~ d_m=\frac{\sum_{i=1}^m\nabla f(x+\sqrt{1-t}Z_i)}{m},\\
&e=\Ebb_{Z} f(x+\sqrt{1-t}Z),~ e_m=\frac{\sum_{i=1}^m f(x+\sqrt{1-t}Z_i)}{m},\\
&d_m^{\prime}=\frac{\sum_{i=1}^m\nabla f(x+\sqrt{1-t}Z_i^{\prime})}{m},~e_m^{\prime}=\frac{\sum_{i=1}^m f(x+\sqrt{1-t}Z_i^{\prime})}{m}.
\end{align*}
Due to $d-d_m=\Ebb\left[d_m^{\prime}-d_m|\bZ\right]$,
then $\|d-d_m\|^2_2\leq \Ebb\left[\|d_m^{\prime}-d_m\|^2_2|\bZ\right]$.
Then,
\begin{align}\label{mm1}
\Ebb\|d-d_m\|^2
&\leq \Ebb\left[\Ebb[\|d_m^{\prime}-d_m\|^2_2|\bZ]\right]
=\Ebb\|d_m^{\prime}-d_m\|^2_2\notag\\
&=\frac{\Ebb_{Z_1,Z_1^{\prime}}\left\|\nabla f(x+\sqrt{1-t}Z_1)-\nabla f(x+\sqrt{1-t}Z_1^{\prime})\right\|^2_2}{m}\notag\\
&\leq \frac{(1-t)\gamma^2}{m}\Ebb_{Z_1,Z_1^{\prime}}\left\|Z_1-Z_1^{\prime}\right\|^2_2\notag\\
&\leq \frac{2p\gamma^2}{m},
\end{align}
where the second inequality holds by (\textbf{A1}). 
Similarly, we also have
\begin{align}\label{mm2}
\Ebb|e-e_m|^2
&\leq \Ebb|e_m^{\prime}-e_m|^2\notag\\
&=\frac{\Ebb_{Z_1,Z_1^{\prime}}
\left|f(x+\sqrt{1-t}Z_1)-f(x+\sqrt{1-t}Z_1^{\prime})\right|^2}{m}\notag\\
&\leq \frac{(1-t)\gamma^2}{m}\Ebb_{Z_1,Z_1^{\prime}}\left\|Z_1-Z_1^{\prime}\right\|^2_2\notag\\
&\leq \frac{2p\gamma^2}{m},
\end{align}
where the second inequality holds due to (\textbf{A1}). 
Thus, by (\ref{mm1}) and (\ref{mm2}), it follows that
\begin{align}\label{mm3}
\underset{x\in \mathbb{R}^p,t\in [0,1]}{\sup}\Ebb\left\|d-d_m\right\|_2^2
\leq \frac{2p\gamma^{2}}{m},
\end{align}
\begin{align}\label{mm4}
\underset{x\in \mathbb{R}^p,t\in [0,1]}{\sup}\Ebb|e-e_m|^2\leq \frac{2p\gamma^2}{m}.
\end{align}
Then, by (\textbf{A1}) and (\textbf{A2}),
through some simple calculation, it yields that
\begin{align}\label{mm5}
\|b(x,t)-\tilde{b}_m(x,t)\|_2
&=\left\|\frac{d}{e}-\frac{d_m}{e_m}\right\|_2\notag\\
&\leq \frac{\|d\|_2|e_m-e|+\|d-d_m\|_2|e|}{|ee_m|}\notag\\
&\leq \frac{\gamma|e_m-e|+\|d-d_m\|_2|e|}{\xi^2}.
\end{align}
Let $R>0$, then
\begin{align}\label{mm6}
\underset{\|x\|_2\leq R}\sup f(x)\leq \mathcal{O}\left(\exp(R^2/2)\right).
\end{align}
Therefore, by (\ref{mm3})-(\ref{mm6}), it can be concluded  that
\begin{align*}
\underset{\|x\|_2 \leq R,t\in [0,1]}{\sup} \Ebb\left[\|b(x,t)-\tilde{b}_m(x,t)\|_2^2\right]
\leq \mathcal{O}\left(\frac{p\exp(R^2)}{m}\right).
\end{align*}

Moreover, if $f$ has the finite upper bound, that is, there exists one finite and positive constant $\zeta$ such that $f\leq\zeta$. Then, similar to (\ref{mm5}), it follows that for all $x\in \mathbb{R}^p$ and $t\in[0,1]$,
\begin{align}\label{mm55}
\|b(x,t)-\tilde{b}_m(x,t)\|_2^2
&\leq  2\frac{\gamma^2|e_m-e|^2+\zeta^2\|d-d_m\|_2^2}{\xi^4}.
\end{align}
Then, by (\ref{mm3})-(\ref{mm4}) and (\ref{mm55}), it follows that
\begin{align*}
\underset{x\in\mathbb{R}^p,t\in [0,1]}{\sup} \mathbb{E}\left[\|b(x,t)-\tilde{b}_m(t,x)\|_2^2\right]
\leq \mathcal{O}\left(\frac{p}{m}\right).
\end{align*}
\end{proof}
\begin{lemma}\label{lemma5}
Assume (\textbf{A1}) and (\textbf{A2}) hold,
then for $k=0,1,\ldots,K$,
\begin{align*}
E[\|\widetilde{Y}_{t_{k}}\|^2_2]
\leq  \frac{6\gamma^2}{\xi^2}+3p.
\end{align*}
\end{lemma}
\begin{proof}
Define
$\Theta_{k,t}=\widetilde{Y}_{t_k}+(t-t_k)\tilde{b}_m(\widetilde{Y}_{t_k},t_k)$
and $\widetilde{Y}_{t}=\Theta_{k,t}+B_t-B_{t_k}$, where $t_k \leq t \leq t_{k+1}$ with $k=0,1,\ldots,K-1$.
By (\textbf{A1}) and (\textbf{A2}), it follows that for all $x \in \mathbb{R}^p$ and $t\in[0,1]$,
\begin{align}\label{eq2}
\|b(x,t)\|_2^2\leq \frac{\gamma^2}{\xi^2},~~ \|\tilde{b}_m(x,t)\|_2^2 \leq \frac{\gamma^2}{\xi^2}.
\end{align}
Then, by (\ref{eq2}),  we have
\begin{align*}
\|\Theta_{k,t}\|^2_2&=\|\widetilde{Y}_{t_k}\|^2_2+(t-t_k)^2\|\tilde{b}_m(\widetilde{Y}_{t_k},t_k)\|^2_2
+2(t-t_k)\widetilde{Y}_{t_k}^{\top}\tilde{b}_m(\widetilde{Y}_{t_k},t_k)\\
&\leq (1+s)\|\widetilde{Y}_{t_k}\|^2_2+\frac{(s+s^2)\gamma^2}{\xi^2}.
\end{align*}
Further, we can get
\begin{align*}
\Ebb[\|\widetilde{Y}_{t}\|^2_2|\widetilde{Y}_{t_k}]&=\Ebb[\|\Theta_{k,t}\|^2_2|\widetilde{Y}_{t_k}]+(t-t_k)p\\
&\leq(1+s)\Ebb\|\widetilde{Y}_{t_k}\|^2_2+\frac{(s+s^2)\gamma^2}{\xi^2}+sp.
\end{align*}
Therefore,
\begin{align*}
\Ebb[\|\widetilde{Y}_{t_{k+1}}\|^2_2]
\leq(1+s)\Ebb\|\widetilde{Y}_{t_k}\|^2_2+\frac{(s+s^2)\gamma^2}{\xi^2}+sp.
\end{align*}
Since $\widetilde{Y}_{t_0}=0$, then by induction, we have
\begin{align*}
\Ebb[\|\widetilde{Y}_{t_{k+1}}\|^2_2]
\leq \frac{6\gamma^2}{\xi^2}+3p.
\end{align*}
\end{proof}
\begin{lemma}\label{lemma6}
Assume (\textbf{A1}) and (\textbf{A2}) hold,
then for $k=0,1,\ldots,K$,
\begin{align*}
\Ebb\left\|b(\widetilde{Y}_{t_k},t_k)-\tilde{b}_m(\widetilde{Y}_{t_k},t_k)\right\|_2^2\leq
\mathcal{O}\left(\frac{p}{\log(m)}\right).
\end{align*}
Moreover, if $f$  has the finite upper bound, then 
\begin{align*}
\Ebb\left\|b(\widetilde{Y}_{t_k},t_k)-\tilde{b}_m(\widetilde{Y}_{t_k},t_k)\right\|_2^2\leq
\mathcal{O}\left(\frac{p}{m}\right).
\end{align*}
\end{lemma}
\begin{proof}
Let $R>0$, then
\begin{equation}\label{erro1}
\begin{split}
\Ebb\left\|b(\widetilde{Y}_{t_k},t_k)-\tilde{b}_m(\widetilde{Y}_{t_k},t_k)\right\|_2^2
&=\Ebb_{\widetilde{Y}_{t_k}}\Ebb_Z\left[\left\|b(\widetilde{Y}_{t_k},t_k)-\tilde{b}_m(\widetilde{Y}_{t_k},t_k)\right\|_2^21(\|\widetilde{Y}_{t_k}\|_2\leq R)\right]\\
&~~~+\Ebb_{\widetilde{Y}_{t_k}}\Ebb_Z\left[\left\|b(\widetilde{Y}_{t_k},t_k)-\tilde{b}_m(\widetilde{Y}_{t_k},t_k)\right\|_2^21(\|\widetilde{Y}_{t_k}\|_2> R)\right].
\end{split}
\end{equation}
Next, we need to bound the two terms of (\ref{erro1}).
First, by Lemma \ref{lemma4}, we have
$$
\Ebb_{\widetilde{Y}_{t_k}}\Ebb_Z\left[\left\|b(\widetilde{Y}_{t_k},t_k)-\tilde{b}_m(\widetilde{Y}_{t_k},t_k)\right\|_2^21(\|\widetilde{Y}_{t_k}\|_2\leq R)\right]\leq  \mathcal{O}\left(\frac{p\exp(R^2)}{m}\right).
$$
Secondly, combining (\ref{eq2}) and Lemma \ref{lemma5} with Markov inequality, it yields that
$$
\Ebb_{\widetilde{Y}_{t_k}}\Ebb_Z\left[\left\|b(\widetilde{Y}_{t_k},t_k)-\tilde{b}_m(\widetilde{Y}_{t_k},t_k)\right\|_2^21(\|\widetilde{Y}_{t_k}\|_2> R)\right]
\leq \mathcal{O}\left(p/R^2\right).
$$
Thence
\begin{align}\label{err1}
\Ebb\left\|b(\widetilde{Y}_{t_k},t_k)-\tilde{b}_m(\widetilde{Y}_{t_k},t_k)\right\|_2^2
\leq \mathcal{O}\left(\frac{p\exp(R^2)}{m}\right)
+\mathcal{O}\left(p/R^2\right).
\end{align}
Set $R=\left(\frac{\log(m)}{2}\right)^{1/2}$ in (\ref{err1}), then we have
\begin{align*}
\Ebb\left\|b(\widetilde{Y}_{t_k},t_k)-\tilde{b}_m(\widetilde{Y}_{t_k},t_k)\right\|_2^2\leq
\mathcal{O}\left(\frac{p}{\log(m)}\right).
\end{align*}

Moreover, if $f$ has the finite upper bound, then by Lemma \ref{lemma4}, we can similarly get
$$
\Ebb\left\|b(\widetilde{Y}_{t_k},t_k)-\tilde{b}_m(\widetilde{Y}_{t_k},t_k)\right\|_2^2=
\Ebb_{\widetilde{Y}_{t_k}}\Ebb_{Z}\left[\left\|b(\widetilde{Y}_{t_k},t_k)-\tilde{b}_m(\widetilde{Y}_{t_k},t_k)\right\|_2^2\right]
\leq \mathcal{O}\left(\frac{p}{m}\right).
$$
This completes the proof.
\end{proof}
\subsection{Proof of Theorem \ref{th1}}
\begin{proof}
From the definition of $\widetilde{Y}_{t_k}$ and $X_{t_k}$, we have
\begin{align*}
&\|\widetilde{Y}_{t_k}-X_{t_k}\|_2^2\\
&\leq \|\widetilde{Y}_{t_{k-1}}-X_{t_{k-1}}\|_2^2
+\left(\int_{t_{k-1}}^{t_k}\|b(X_u,u)-\tilde{b}_m(\widetilde{Y}_{t_{k-1}},t_{k-1})\|_2\mathrm{d} u\right)^2\\
&~~~+2\|\widetilde{Y}_{t_{k-1}}-X_{t_{k-1}}\|_2
\left(\int_{t_{k-1}}^{t_k}\|b(X_u,u)-\tilde{b}_m(\widetilde{Y}_{t_{k-1}},t_{k-1})\|_2\mathrm{d} u\right)\\
&\leq (1+s) \|\widetilde{Y}_{t_{k-1}}-X_{t_{k-1}}\|_2^2
+(1+s)\int_{t_{k-1}}^{t_k}\|b(X_u,u)-\tilde{b}_m(\widetilde{Y}_{t_{k-1}},t_{k-1})\|_2^2\mathrm{d} u\\
&\leq (1+s) \|\widetilde{Y}_{t_{k-1}}-X_{t_{k-1}}\|_2^2
+2(1+s)\int_{t_{k-1}}^{t_k}\|b(X_u,u)-b(\widetilde{Y}_{t_{k-1}},t_{k-1})\|_2^2\mathrm{d} u\\
&~~~+2s(1+s)\|b(\widetilde{Y}_{t_{k-1}},t_{k-1})-\tilde{b}_m(\widetilde{Y}_{t_{k-1}},t_{k-1})\|_2^2\\
&\leq (1+s) \|\widetilde{Y}_{t_{k-1}}-X_{t_{k-1}}\|_2^2+4C_1^2(1+s)\int_{t_{k-1}}^{t_k}[\|X_u-\widetilde{Y}_{t_{k-1}}\|_2^2+|u-t_{k-1}|]\mathrm{d} u\\
&~~~+2s(1+s)\|b(\widetilde{Y}_{t_{k-1}},t_{k-1})-\tilde{b}_m(\widetilde{Y}_{t_{k-1}},t_{k-1})\|_2^2\\
&\leq (1+s) \|\widetilde{Y}_{t_{k-1}}-X_{t_{k-1}}\|_2^2
+8C_1^2(1+s)\int_{t_{k-1}}^{t_k}\|X_u-X_{t_{k-1}}\|_2^2\mathrm{d} u\\
&~~~+8C_1^2s(1+s)\|X_{t_{k-1}}-\widetilde{Y}_{t_{k-1}}\|_2^2+4C_1^2(1+s)s^2\\
&~~~+2s(1+s)\|b(\widetilde{Y}_{t_{k-1}},t_{k-1})-\tilde{b}_m(\widetilde{Y}_{t_{k-1}},t_{k-1})\|_2^2\\
&\leq (1+s+8C_1^2(s+s^2))\|\widetilde{Y}_{t_{k-1}}-X_{t_{k-1}}\|_2^2
+8C_1^2(1+s)\int_{t_{k-1}}^{t_k}\|X_u-X_{t_{k-1}}\|_2^2\mathrm{d} u\\
&~~~+4C_1^2(1+s)s^2+2s(1+s)\|b(\widetilde{Y}_{t_{k-1}},t_{k-1})-\tilde{b}_m(\widetilde{Y}_{t_{k-1}},t_{k-1})\|_2^2,
\end{align*}
where the second inequality holds duo to $2ab\leq s a^2+\frac{b^2}{s}$,
the fourth inequality holds by (\ref{cond3}).
Then,
\begin{align}\label{th1eq1}
\Ebb\|\widetilde{Y}_{t_k}-X_{t_k}\|_2^2
&\leq (1+s+8C_1^2(s+s^2)) \Ebb\|\widetilde{Y}_{t_{k-1}}-X_{t_{k-1}}\|_2^2\notag\\
&~~~+8C_1^2(1+s)\int_{t_{k-1}}^{t_k}\Ebb\|X_u-X_{t_{k-1}}\|_2^2\mathrm{d} u+4C_1^2(s^2+s^3)\notag\\
&~~~+2s(1+s)\Ebb[\|b(\widetilde{Y}_{t_{k-1}},t_{k-1})-\tilde{b}_m(\widetilde{Y}_{t_{k-1}},t_{k-1})\|_2^2]\notag\\
&\leq(1+s+8C_1^2(s+s^2)) \Ebb\|\widetilde{Y}_{t_{k-1}}-X_{t_{k-1}}\|_2^2+h(s)\notag\\
&~~~+4C_1^2(s^2+s^3)+2s(1+s)\Ebb[\|b(\widetilde{Y}_{t_{k-1}},t_{k-1})-\tilde{b}_m(\widetilde{Y}_{t_{k-1}},t_{k-1})\|_2^2]\notag\\
&\leq(1+s+8C_1^2(s+s^2)) \Ebb\|\widetilde{Y}_{t_{k-1}}-X_{t_{k-1}}\|_2^2+h(s)\notag\\
&~~~+4C_1^2(s^2+s^3)+2s(1+s)\mathcal{O}\left(\frac{p}{\log(m)}\right),
\end{align}
where $h(s)=8C_1^2(s+s^2)[4C_0\exp(2C_0)(C_0+p)s^2+2C_0s^2+2ps]$, and the last inequality holds by Lemma \ref{lemma6}.
Owing to $\widetilde{Y}_{t_0}=X_{t_0}=0$, we can conclude that
\begin{align*}
&\Ebb\|\widetilde{Y}_{t_K}-X_{t_K}\|_2^2\\
&\leq\frac{(1+s+8C_1^2(s+s^2))^K-1}{s+8C_1^2(s+s^2)}
\left[h(s)+4C_1^2(s^2+s^3)+2(s+s^2)\mathcal{O}\left(\frac{p}{\log(m)}\right)\right]\\
&\leq \mathcal{O}(ps)+\mathcal{O}\left(\frac{p}{\log(m)}\right).
\end{align*}
Therefore,
\begin{align*}
W_2(Law(\widetilde{Y}_{t_K}),\mu)
\leq \mathcal{O}(\sqrt{ps})+\mathcal{O}\left(\sqrt{\frac{p}{\log(m)}}\right).
\end{align*}
\end{proof}
\subsection{Proof of Theorem \ref{th2}}
\begin{proof}
This proof is same as that of Theorem \ref{th1}.  Similar to (\ref{th1eq1}), by Lemma \ref{lemma6}, it yields that
\begin{align*}
\Ebb\|\widetilde{Y}_{t_k}-X_{t_k}\|_2^2
&\leq(1+s+8C_1^2(s+s^2)) \Ebb\|\widetilde{Y}_{t_{k-1}}-X_{t_{k-1}}\|_2^2+h(s)\\
&~~~+4C_1^2(s^2+s^3)+2s(1+s)\mathcal{O}\left(\frac{p}{m}\right).
\end{align*}
Then,  we also have
\begin{align*}
&\Ebb\|\widetilde{Y}_{t_K}-X_{t_K}\|_2^2\\
&\leq\frac{(1+s+8C_1^2(s+s^2))^K-1}{s+8C_1^2(s+s^2)}
\left[h(s)+4C_1^2(s^2+s^3)+2(s+s^2)\mathcal{O}\left(\frac{1}{m}\right)\right]\\
&\leq \mathcal{O}(ps)+\mathcal{O}\left(\frac{p}{m}\right).
\end{align*}
Hence, it follows that
\begin{align*}
W_2(Law(\widetilde{Y}_{t_K}),\mu)
\leq \mathcal{O}(\sqrt{ps})+\mathcal{O}\left(\sqrt{\frac{p}{m}}\right).
\end{align*}
\end{proof}

\subsection{Preliminary lemmas for Theorem \ref{th3}}
To prove Theorem \ref{th3}, we first prove the Lemmas \ref{lemma41}-\ref{lemma61}.
\begin{lemma}\label{lemma41}
Assume (\textbf{A1}) holds,
then for any $R>0$,
\begin{align*}
\underset{\|x\|_2 \leq R,t\in [0,1]}{\sup} \Ebb\left[\|b(x,t)-\tilde{b}_m(x,t)\|_2^2\right]
\leq
\mathcal{O}\left(\frac{p\exp(R^2)(C_p)^4}{m\varepsilon^4}\right)
+\mathcal{O}\left(\frac{p(C_p)^2}{m\varepsilon^2}\right),
\end{align*}
where $C_p=(2\pi)^{p/2}C^{-1}$.
Moreover, if $f$ has the finite upper bound, then
\begin{align*}
\underset{x\in\mathbb{R}^p,t\in [0,1]}{\sup} \Ebb\left[\|b(x,t)-\tilde{b}_m(t,x)\|_2^2\right]
\leq \mathcal{O}\left(\frac{p(C_p)^4}{m\varepsilon^4}\right)
+\mathcal{O}\left(\frac{p(C_p)^2}{m\varepsilon^2}\right).
\end{align*}
\end{lemma}
\begin{proof}
Denote two independent sets of independent copies of $Z\sim N(0, \bI_p)$ by $\bZ=\{Z_1,\ldots,Z_m\}$ and $\bZ^{\prime}=\{Z_1^{\prime},\ldots,Z_m^{\prime}\}$.
For notation convenience, we  denote
\begin{align*}
&d=\Ebb_{Z}\nabla g(x+\sqrt{1-t}Z),~ d_m=\frac{\sum_{i=1}^m\nabla g(x+\sqrt{1-t}Z_i)}{m},\\
&e=\Ebb_{Z} \left[g(x+\sqrt{1-t}Z)+\frac{\varepsilon }{ C_p(1-\varepsilon)}\right],~ e_m=\frac{\sum_{i=1}^m g(x+\sqrt{1-t}Z_i)}{m}+\frac{\varepsilon}{ C_p(1-\varepsilon)},\\
&d_m^{\prime}=\frac{\sum_{i=1}^m\nabla g(x+\sqrt{1-t}Z_i^{\prime})}{m},~e_m^{\prime}=\frac{\sum_{i=1}^m g(x+\sqrt{1-t}Z_i^{\prime})}{m}+\frac{\varepsilon}{ C_p(1-\varepsilon)},
\end{align*}
where $g(x)=\exp(\|x\|^2_2/2-V(x))$.
Since  $d-d_m=\Ebb\left[d_m^{\prime}-d_m|\bZ\right]$,
we have  $\|d-d_m\|^2_2\leq \Ebb\left[\|d_m^{\prime}-d_m\|^2_2|\bZ\right]$.
By (\textbf{A1}),
it yields that $g$ and $\nabla g$ are Lipschitz continuous.
Thus  there exists a finite and positive constant $\gamma$
such that for all $x,y \in \mathbb{R}^p$,
\begin{align}\label{L11}
|g(x)-g(y)|\leq \gamma \|x-y\|_2,
\end{align}
\begin{align}\label{L21}
\|\nabla g(x)- \nabla g(y)\|_2\leq \gamma \|x-y\|_2.
\end{align}
Therefore,
\begin{align}\label{mm11}
\Ebb\|d-d_m\|^2
&\leq \Ebb\left[\Ebb[\|d_m^{\prime}-d_m\|^2_2|\bZ]\right]=\Ebb\|d_m^{\prime}-d_m\|^2_2\notag\\
&=\frac{\Ebb_{Z_1,Z_1^{\prime}}\left\|\nabla g(x+\sqrt{1-t}Z_1)-\nabla g(x+\sqrt{1-t}Z_1^{\prime})\right\|^2_2}{m}\notag\\
&\leq \frac{(1-t)\gamma^2}{m}\Ebb_{Z_1,Z_1^{\prime}}\left\|Z_1-Z_1^{\prime}\right\|^2_2\notag\\
&\leq \frac{2p\gamma^2}{m},
\end{align}
where the second inequality follows from (\ref{L21}).
Similarly, we also have
\begin{align}\label{mm21}
\Ebb|e-e_m|^2
&\leq \Ebb|e_m^{\prime}-e_m|^2\notag\\
&=\frac{\Ebb_{Z_1,Z_1^{\prime}}
\left|g(x+\sqrt{1-t}Z_1)-g(x+\sqrt{1-t}Z_1^{\prime})\right|^2}{m}\notag\\
&\leq \frac{(1-t)\gamma^2}{m}\Ebb_{Z_1,Z_1^{\prime}}\left\|Z_1-Z_1^{\prime}\right\|^2_2\notag\\
&\leq \frac{2p\gamma^2}{m},
\end{align}
where the second inequality follows from (\ref{L11}).
Hence, by (\ref{mm11}) and (\ref{mm21}), we have
\begin{align}\label{mm31}
\underset{x\in \mathbb{R}^p,t\in [0,1]}{\sup}\Ebb\left\|d-d_m\right\|_2^2
\leq \frac{2p\gamma^{2}}{m},
\end{align}
\begin{align}\label{mm41}
\underset{x\in \mathbb{R}^p,t\in [0,1]}{\sup}\Ebb|e-e_m|^2\leq \frac{2p\gamma^2}{m}.
\end{align}
Then, by (\ref{L11}) and (\ref{L21}), 
through some simple calculation, it yields that
\begin{align}\label{mm51}
\|b(x,t)-\tilde{b}_m(x,t)\|_2
&=\left\|\frac{d}{e}-\frac{d_m}{e_m}\right\|_2\notag\\
&\leq \frac{\|d\|_2|e_m-e|+\|d-d_m\|_2|e|}{|ee_m|}\notag\\
&\leq \frac{\gamma|e_m-e|+\|d-d_m\|_2|e|}{(\varepsilon /( C_p-C_p\varepsilon))^2}.
\end{align}
Let $R>0$, then
\begin{align}\label{mm61}
\underset{\|x\|_2\leq R}\sup g(x)\leq \mathcal{O}\left(\exp(R^2/2)\right).
\end{align}
Therefore, by (\ref{mm31})-(\ref{mm61}), it can be concluded  that
\begin{align*}
\underset{\|x\|_2 \leq R,t\in [0,1]}{\sup} \Ebb\left[\|b(x,t)-\tilde{b}_m(x,t)\|_2^2\right]
&\leq
\mathcal{O}\left(\frac{p\exp(R^2)}{m(\varepsilon /(C_p-C_p\varepsilon))^4}\right)
+
\mathcal{O}\left(\frac{p}{m(\varepsilon /(C_p-C_p\varepsilon))^2}\right)\\
&\leq
\mathcal{O}\left(\frac{p\exp(R^2)(C_p)^4}{m\varepsilon^4}\right)
+
\mathcal{O}\left(\frac{p(C_p)^2}{m\varepsilon^2}\right).
\end{align*}

Moreover,  $f$ has a finite upper bound so does $g$. Then there exists a finite and positive constant $\zeta$ such that $g\leq\zeta$. Similar to (\ref{mm51}), it follows that for all $x\in \mathbb{R}^p$ and $t\in[0,1]$,
\begin{align}\label{mm551}
\|b(x,t)-\tilde{b}_m(x,t)\|_2^2
&\leq  2\frac{\gamma^2|e_m-e|^2+(\zeta+\varepsilon /(C_p-C_p\varepsilon))^2\|d-d_m\|_2^2}{(\varepsilon /(C_p-C_p\varepsilon))^4}.
\end{align}
Then, by (\ref{mm31})-(\ref{mm41}) and (\ref{mm551}), it follows that
\begin{align*}
\underset{x\in\mathbb{R}^p,t\in [0,1]}{\sup} \Ebb\left[\|b(x,t)-\tilde{b}_m(t,x)\|_2^2\right]
&\leq \mathcal{O}\left(\frac{p}{m(\varepsilon /(C_p-C_p\varepsilon))^4}\right)+\mathcal{O}\left(\frac{p}{m(\varepsilon /(C_p-C_p\varepsilon))^2}\right)\\
&\leq \mathcal{O}\left(\frac{p(C_p)^4}{m\varepsilon^4}\right)+\mathcal{O}\left(\frac{p(C_p)^2}{m\varepsilon^2}\right).
\end{align*}
\end{proof}
\begin{lemma}\label{lemma51}
Assume (\textbf{A1}) holds,
then for $k=0,1,\ldots,K$,
\begin{align*}
\Ebb[\|\widetilde{Y}_{t_{k}}^{\varepsilon}\|^2_2]
\leq\mathcal{O}\left(\frac{(C_p)^2}{\varepsilon^2}\right)+\mathcal{O}\left(p\right),
\end{align*}
where $C_p=(2\pi)^{p/2}C^{-1}$.
\end{lemma}
\begin{proof}
Define
$\Theta_{k,t}^{\varepsilon}=\widetilde{Y}_{t_k}^{\varepsilon}+(t-t_k)\tilde{b}_m(\widetilde{Y}_{t_k}^{\varepsilon},t_k)$
and $\widetilde{Y}_{t}^{\varepsilon}=\Theta_{k,t}^{\varepsilon}+B_t-B_{t_k}$, where $t_k \leq t \leq t_{k+1}$ with $k=0,1,\ldots,K-1$.
By (\textbf{A1}), then there exists one finite and positive constant $\gamma$ such that  $g$ is $\gamma$-Lipschitz continuous.
Then, for all $x \in \mathbb{R}^p$ and $t\in[0,1]$, we have
\begin{align}\label{eq21}
\|b(x,t)\|_2^2\leq \frac{\gamma^2}{(\varepsilon /(C_p-C_p\varepsilon))^2
},~~ \|\tilde{b}_m(x,t)\|_2^2 \leq \frac{\gamma^2}{(\varepsilon /(C_p-C_p\varepsilon))^2
}.
\end{align}
By (\ref{eq21}),  we have
\begin{align*}
\|\Theta_{k,t}^{\varepsilon}\|^2_2&=\|\widetilde{Y}_{t_k}^{\varepsilon}\|^2_2+(t-t_k)^2\|\tilde{b}_m(\widetilde{Y}_{t_k}^{\varepsilon},t_k)\|^2_2
+2(t-t_k)(\widetilde{Y}_{t_k}^{\varepsilon})^{\top}\tilde{b}_m(\widetilde{Y}_{t_k}^{\varepsilon},t_k)\\
&\leq (1+s)\|\widetilde{Y}_{t_k}^{\varepsilon}\|^2_2+\frac{(s+s^2)\gamma^2}{(\varepsilon /(C_p-C_p\varepsilon))^2}.
\end{align*}
Furthermore, it can be shown that
\begin{align*}
\Ebb[\|\widetilde{Y}_{t}^{\varepsilon}\|^2_2|\widetilde{Y}_{t_k}^{\varepsilon}]
&=\Ebb[\|\Theta_{k,t}^{\varepsilon}\|^2_2|\widetilde{Y}_{t_k}^{\varepsilon}]+(t-t_k)p\\
&\leq(1+s)\Ebb\|\widetilde{Y}_{t_k}^{\varepsilon}\|^2_2+\frac{(s+s^2)\gamma^2}{(\varepsilon /(C_p-C_p\varepsilon))^2}+sp.
\end{align*}
Therefore,
\begin{align*}
\Ebb[\|\widetilde{Y}_{t_{k+1}}^{\varepsilon}\|^2_2]
\leq(1+s)\Ebb[\|\widetilde{Y}_{t_k}^{\varepsilon}\|^2_2]+\frac{(s+s^2)\gamma^2}{(\varepsilon /(C_p-C_p\varepsilon))^2}+sp.
\end{align*}
Since $\widetilde{Y}_{t_0}^{\varepsilon}=0$, then by induction, we have
\begin{align*}
\Ebb[\|\widetilde{Y}_{t_{k+1}}^{\varepsilon}\|^2_2]
\leq \frac{6\gamma^2}{(\varepsilon /(C_p-C_p\varepsilon))^2}+3p
\leq\mathcal{O}\left(\frac{(C_p)^2}{\varepsilon^2}\right)+\mathcal{O}\left(p\right).
\end{align*}
\end{proof}
\begin{lemma}\label{lemma61}
Assume (\textbf{A1})  holds,
then for $k=0,1,\ldots,K$ and $t\in [0,1]$,
\begin{align*}
\Ebb\left\|b(\widetilde{Y}_{t_k}^{\varepsilon},t_k)-\tilde{b}_m(\widetilde{Y}_{t_k}^{\varepsilon},t_k)\right\|_2^2
&\leq
\mathcal{O}\left(\frac{p(C_p)^4}{\sqrt{m}\varepsilon^4}\right)
+\mathcal{O}\left(\frac{(C_p)^4}{\log(m)\varepsilon^4}\right)
+\mathcal{O}\left(\frac{p(C_p)^2}{\log(m)\varepsilon^2}\right),
\end{align*}
where $C_p=(2\pi)^{p/2}C^{-1}$.
Moreover, if $f$  has a finite upper bound, then
\begin{align*}
\Ebb\left\|b(\widetilde{Y}_{t_k}^{\varepsilon},t_k)-\tilde{b}_m(\widetilde{Y}_{t_k}^{\varepsilon},t_k)\right\|_2^2
 \leq \mathcal{O}\left(\frac{p(C_p)^4}{m\varepsilon^4}\right)
+\mathcal{O}\left(\frac{p(C_p)^2}{m\varepsilon^2}\right).
\end{align*}
\end{lemma}
\begin{proof}
Let $R>0$, then
\begin{equation}\label{erro11}
\begin{split}
\Ebb\left\|b(\widetilde{Y}_{t_k}^{\varepsilon},t_k)-\tilde{b}_m(\widetilde{Y}_{t_k}^{\varepsilon},t_k)\right\|_2^2
&=\Ebb_{\widetilde{Y}_{t_k}^{\varepsilon}}\Ebb_Z\left[\left\|b(\widetilde{Y}_{t_k}^{\varepsilon},t_k)-\tilde{b}_m(\widetilde{Y}_{t_k}^{\varepsilon},t_k)\right\|_2^21(\|\widetilde{Y}_{t_k}^{\varepsilon}\|_2\leq R)\right]\\
&~~~+\Ebb_{\widetilde{Y}_{t_k}^{\varepsilon}}\Ebb_Z\left[\left\|b(\widetilde{Y}_{t_k}^{\varepsilon},t_k)-\tilde{b}_m(\widetilde{Y}_{t_k}^{\varepsilon},t_k)\right\|_2^21(\|\widetilde{Y}_{t_k}^{\varepsilon}\|_2> R)\right].
\end{split}
\end{equation}
Next, we need to bound the two terms on the right hand of (\ref{erro11}).
First, by Lemma \ref{lemma41}, we have
\begin{align*}
\Ebb_{\widetilde{Y}_{t_k}^{\varepsilon}}\Ebb_Z\left[\left\|b(\widetilde{Y}_{t_k}^{\varepsilon},t_k)-\tilde{b}_m(\widetilde{Y}_{t_k}^{\varepsilon},t_k)\right\|_2^21(\|\widetilde{Y}_{t_k}^{\varepsilon}\|_2\leq R)\right]
\leq
\mathcal{O}\left(\frac{p\exp(R^2)(C_p)^4}{m\varepsilon^4}\right)
+\mathcal{O}\left(\frac{p(C_p)^2}{m\varepsilon^2}\right).
\end{align*}
Second, by combining (\ref{eq21}) and Lemma \ref{lemma51} with the Markov inequality, we have
\begin{align*}
\Ebb_{\widetilde{Y}_{t_k}^{\varepsilon}}\Ebb_Z\left[\left\|b(\widetilde{Y}_{t_k}^{\varepsilon},t_k)-\tilde{b}_m(\widetilde{Y}_{t_k}^{\varepsilon},t_k)\right\|_2^21(\|\widetilde{Y}_{t_k}^{\varepsilon}\|_2> R)\right]
&\leq \mathcal{O}\left(\frac{(C_p)^4}{R^2\varepsilon^4}\right)+\mathcal{O}\left(\frac{p(C_p)^2}{R^2\varepsilon^2}\right).
\end{align*}
Therefore,
\begin{align}\label{err11}
\Ebb\left\|b(\widetilde{Y}_{t_k}^{\varepsilon},t_k)-\tilde{b}_m(\widetilde{Y}_{t_k}^{\varepsilon},t_k)\right\|_2^2
&\leq
\mathcal{O}\left(\frac{p\exp(R^2)(C_p)^4}{m\varepsilon^4}\right)
+\mathcal{O}\left(\frac{p(C_p)^2}{m\varepsilon^2}\right)\notag\\
&~~~+\mathcal{O}\left(\frac{(C_p)^4}{R^2\varepsilon^4}\right)
+\mathcal{O}\left(\frac{p(C_p)^2}{R^2\varepsilon^2}\right).
\end{align}
Setting $R=\left(\frac{\log(m)}{2}\right)^{1/2}$ in (\ref{err11}), we have
\begin{align*}
\Ebb\left\|b(\widetilde{Y}_{t_k}^{\varepsilon},t_k)-\tilde{b}_m(\widetilde{Y}_{t_k}^{\varepsilon},t_k)\right\|_2^2
&\leq
\mathcal{O}\left(\frac{p(C_p)^4}{\sqrt{m}\varepsilon^4}\right)
+\mathcal{O}\left(\frac{(C_p)^4}{\log(m)\varepsilon^4}\right)
+\mathcal{O}\left(\frac{p(C_p)^2}{\log(m)\varepsilon^2}\right).
\end{align*}
Moreover, if $f$ has a finite upper bound, then by Lemma \ref{lemma41}, we can similarly get
\begin{align*}
\Ebb\left\|b(\widetilde{Y}_{t_k}^{\varepsilon},t_k)-\tilde{b}_m(\widetilde{Y}_{t_k}^{\varepsilon},t_k)\right\|_2^2
&=\Ebb_{\widetilde{Y}_{t_k}^{\varepsilon}}\Ebb_{Z}\left[\left\|b(\widetilde{Y}_{t_k}^{\varepsilon},t_k)-\tilde{b}_m(\widetilde{Y}_{t_k}^{\varepsilon},t_k)\right\|_2^2\right]\\
& \leq \mathcal{O}\left(\frac{p(C_p)^4}{m\varepsilon^4}\right)
+\mathcal{O}\left(\frac{p(C_p)^2}{m\varepsilon^2}\right).
\end{align*}
This completes the proof.
\end{proof}
\subsection{Proof of Theorem \ref{th3}}
\begin{proof}
By triangle inequality, we have  $W_2(Law(\widetilde{Y}_{t_K}^{\varepsilon}),\mu)\leq
W_2(Law(\widetilde{Y}_{t_K}^{\varepsilon}),\mu_{\varepsilon})
+W_2(\mu,\mu_{\varepsilon})$,
then we  obtain the upper bound of two terms on the right hand of this inequality, respectively.

First,
similar to the proof of Theorem \ref{th1}, by Lemma \ref{lemma61} and $\widetilde{Y}_{t_0}^{\varepsilon}=X_{t_0}=0$ and through some calculation, we can conclude that
\begin{align}\label{ERRR1}
W_2(Law(\widetilde{Y}_{t_K}^{\varepsilon}),\mu_{\varepsilon})
\leq
\mathcal{O}(\sqrt{ps})+\mathcal{O}\left(\frac{\sqrt{p}(C_p)^2}{m^{1/4}\varepsilon^2}\right)
+\mathcal{O}\left(\frac{(C_p)^2}{\sqrt{\log(m)}\varepsilon^2}\right)
+\mathcal{O}\left(\frac{\sqrt{p}C_p}{\sqrt{\log(m)}\varepsilon}\right).
\end{align}

Second, we need to get the upper bound of $W_2(\mu,\mu_{\varepsilon})$.
Let $Y\sim \mu$ and $Z\sim N(0,\bI_p)$,  $\theta$ is one Bernoulli random variable satisfying  $P(\theta=1)=1-\varepsilon$ and
$P(\theta=0)=\varepsilon$. Assume $Y$, $Z$ and $\theta$ are independent of each other.
Then $(Y,(1-\theta)Z+\theta Y)$ is one coupling of $(\mu,\mu_{\varepsilon})$, and denote  its joint distribution by $\pi$. Therefore, we have
\begin{align*}
\int_{\mathbb{R}^p\times\mathbb{R}^p}\|x-y\|_2^2d \pi&=\Ebb\left\|Y-((1-\theta)Z+\theta Y)\right\|_2^2\\
&=\Ebb\left[\Ebb\left[\|Y-((1-\theta)Z+\theta Y)\|_2^2|\theta\right]\right]\\
&=\Ebb\left[\Ebb\left[\|Y-((1-\theta)Z+\theta Y)\|_2^2|\theta=1\right]\right]P(\theta=1)\\
&~~~+\Ebb\left[\Ebb\left[\|Y-((1-\theta)Z+\theta Y)\|_2^2|\theta=0\right]\right]P(\theta=0)\\
&=\Ebb[\|Y-Z\|_2^2|\theta=0]P(\theta=0)\\
&=\varepsilon \Ebb\|Y-Z\|_2^2\\
&\leq\mathcal{O}(p\varepsilon).
\end{align*}
Then we have
\begin{align}\label{W2mu}
W_2(\mu,\mu_{\varepsilon})\leq \mathcal{O}(\sqrt{p\varepsilon}).
\end{align}
Combining (\ref{ERRR1}) with (\ref{W2mu}), it yields  that
\begin{align}\label{error1}
&W_2(Law(\widetilde{Y}_{t_K}^{\varepsilon}),\mu)\notag\\
&~~\leq\mathcal{O}(\sqrt{p\varepsilon})+
\mathcal{O}(\sqrt{ps})+\mathcal{O}\left(\frac{\sqrt{p}(C_p)^2}{m^{1/4}\varepsilon^2}\right)+\mathcal{O}\left(\frac{(C_p)^2}{\sqrt{\log(m)}\varepsilon^2}\right)
+\mathcal{O}\left(\frac{\sqrt{p}C_p}{\sqrt{\log(m)}\varepsilon}\right).
\end{align}
Set $\varepsilon=(\log(m))^{-1/5}$ in (\ref{error1}),  then there exist one constant $\widetilde{C}_p$ depending on $p$ such that
\begin{align*}
W_2(Law(\widetilde{Y}_{t_K}^{\varepsilon}),\mu)
\leq\widetilde{C}_p\cdot\mathcal{O}\left(\frac{1}{(\log(m))^{1/10}}\right)+ \mathcal{O}(\sqrt{ps}).
\end{align*}

Moreover, if $f$ has the finite upper bound, then similar to the proof of Theorem \ref{th2} and by (\ref{W2mu}) and Lemma \ref{lemma61}, we have
\begin{align}\label{error2}
W_2(Law(\widetilde{Y}_{t_K}^{\varepsilon}),\mu)
\leq \mathcal{O}(\sqrt{p\varepsilon})+\mathcal{O}(\sqrt{ps})+
\mathcal{O}\left(\frac{\sqrt{p}(C_p)^2}{\sqrt{m}\varepsilon^2}\right)
+\mathcal{O}\left(\frac{\sqrt{p}C_p}{\sqrt{m}\varepsilon}\right).
\end{align}
Set $\varepsilon=m^{-1/5}$ in (\ref{error2}), then there exists one constant $\widetilde{C}_p$ depending on $p$ such that
\begin{align*}
W_2(Law(\widetilde{Y}_{t_K}^{\varepsilon}),\mu)
\leq\widetilde{C}_p\cdot\mathcal{O}\left(\frac{1}{m^{1/10}}\right)+ \mathcal{O}(\sqrt{ps}).
\end{align*}

\end{proof}

\end{appendix}

\bibliographystyle{siam}
\bibliography{ref_bib}

\end{document}